\documentclass[10pt]{article}
\usepackage{amsmath}
\usepackage{amsfonts}
\usepackage{amssymb}
\usepackage{amscd}
\usepackage{amsthm}
\usepackage{latexsym}
\usepackage{mathrsfs}
\usepackage{stmaryrd}
\usepackage{authblk}
\usepackage{fullpage}

\newtheorem{theorem}{Theorem}[section]
\newtheorem{proposition}[theorem]{Proposition}
\newtheorem{lemma}[theorem]{Lemma}
\newtheorem{corollary}[theorem]{Corollary}

\theoremstyle{definition}
\newtheorem{definition}[theorem]{Definition}
\newtheorem{example}[theorem]{Example}

\newcommand{\dom}{d}
\newcommand{\ad}{\mathit{ad}}
\newcommand{\Sup}{\bigvee}
\newcommand{\Inf}{\bigwedge}

\begin{document}

\title{Domain Semirings United}

\author[1]{Uli Fahrenberg}
\author[2]{Christian Johansen} 
\author[3]{Georg Struth} 
\author[4]{Krzysztof Ziemia\'nski}
\affil[1]{{\'E}cole Polytechnique, Palaiseau, France}
\affil[2]{ University of Oslo, Norway} 
\affil[3]{University of Sheffield, UK }
\affil[4]{University of Warsaw, Poland}

\date{}

% \documentclass[notnumberedtheorems]{actacyb}

% % Reveals overfull lines.
% %\draft

% % Provides line numbers for the referee.
% %\referee

% % Draws the border of textbox.
% %\usepackage{showframe}

% \newcommand{\dom}{\textup{\textsf{d}}}
% \newcommand{\ad}{\textup{\textsf{a}}}
% \newcommand{\Sup}{\bigvee}
% \newcommand{\Inf}{\bigwedge}

% \hyphenation{semi-ring semi-rings}

% \begin{document}

% %% Uncomment commands below once the required information is known
% %% Used by the technical editor
% %\acta{17}{2005}
% %\setcounter{page}{163}
% %\received{...} 

% \title{Domain Semirings United}

% \author{Uli Fahrenberg\thanks{{\'E}cole Polytechnique, France,
%     \email{uli@lix.polytechnique.fr}}, Christian
%   Johansen\thanks{Norwegian University of Science and Technology,
%     \email{christian.johansen@ntnu.no}}, Georg
%   Struth\thanks{University of Sheffield, UK,
%     \email{g.struth@sheffield.ac.uk}} \\ \and Krzysztof
%   Ziemia\'nski\thanks{University of Warsaw, Poland,
%     \email{ziemians@mimuw.edu.pl}}}
% %\headingauthor{If list of authors gets too long put a shorter version here}

\maketitle

\begin{abstract}
  Domain operations on semirings have been axiomatised in two
  different ways: by a map from an additively idempotent semiring into
  a boolean subalgebra of the semiring bounded by the additive and
  multiplicative unit of the semiring, or by an endofunction on a
  semiring that induces a distributive lattice bounded by the two
  units as its image.  This note presents classes of semirings where
  these approaches coincide.

\end{abstract}

\pagestyle{plain}

%%%%%%%%%%%%%%%%%%%%%%%%%%%%%%%%%%%%%%%%%%%%%%%%%%%%%%%%%%%%

\section{Introduction}\label{S:introduction}

Domain semirings and Kleene algebras with
domain~\cite{DesharnaisMS06,DesharnaisS11} yield particularly simple
program verification formalisms in the style of dynamic logics,
algebras of predicate transformers or boolean algebras with operators
(which are all related).

There are two kinds of axiomatisation. Both are inspired by properties
of the domain operation on binary relations, but target other
computationally interesting models such as program traces or paths on
digraphs as well.

The initial two-sorted axiomatisation~\cite{DesharnaisMS06} models the
domain operation as a map $\dom:S\to B$ from an additively idempotent
semiring $(S,+,\cdot,0,1)$ into a boolean subalgebra $B$ of $S$
bounded by $0$ and $1$. This seems natural as domain elements form
powerset algebras in the target models mentioned. Yet the domain
algebra $B$ cannot be chosen freely: $B$ must be the maximal boolean
subalgebra of $S$ bounded by $0$ and $1$ and equal to the set $S_\dom$
of fixpoints of $\dom$ in $S$.

The alternative, one-sorted axiomatisation~\cite{DesharnaisS11}
therefore models $\dom$ as an endofunction on a semiring $S$ that
induces a suitable domain algebra on $S_\dom$---yet generally only a
bounded distributive lattice. An antidomain (or domain
complementation) operation is needed to obtain boolean domain
algebras.

In the model of binary relations over a set $X$, $+$ is set union and
$\cdot$ relational composition; $0$ is the empty relation and $1$ the
identity relation. The domain of relation $R\subseteq X\times X$ is
$\dom(R)=\{(x,x)\mid \exists y.\ (x,y)\in R\}$ while its antidomain is
$\ad(R)= \{(x,x)\mid \forall y.\ (x,y) \notin R\}$.  In the path model
over a directed graph $\sigma,\tau:E\to V$, the carrier set consists
of all finite paths $(v_1,e_1,v_2,\dots, v_{n-1},e_{n-1},v_n)$ in the
graph in which vertices $v_i\in V$ and edges $e_i\in E$ alternate and
are compatible with the source map $\sigma$ and target map $\tau$.
The operations $+$ and $0$ are again $\cup$ and $\emptyset$,
respectively; $1$ is $V$ with elements $v\in V$ seen as paths of
length $1$. Extending $\sigma$ and $\tau$ to paths as expected,
composition $\pi_1; \pi_2$ of paths $\pi_1$ and $\pi_2$ is defined if
$\tau(\pi_1)=\sigma(\pi_2)$, and it then glues on this vertex. Path
composition is lifted to sets of paths as
$P;Q=\{\pi_1;\pi_2\mid \pi_1\in P,\pi_2\in Q,
\tau(\pi_1)=\sigma(\pi_2)\}$. Finally
$\dom(P)=\{\sigma(\pi)\mid \pi\in P\}$ and
$\ad(P)=\{v\mid \forall \pi.\ \sigma(\pi)=v\Rightarrow \pi\not\in
P\}$.  Other models can be found in the literature.

This note revisits the two axiomatisations mentioned above to tie some
loose ends together. We describe a natural algebraic setting in which
they coincide, and which has so far been overlooked. It consists of
additively idempotent semirings in which the sets of all elements
below $1$ form boolean algebras, as is the case, for instance, in
boolean monoids and boolean quantales.  We further take the
opportunity to discuss domain axioms for arbitrary quantales.

The restriction to such boolean settings has little impact on
applications: most models of interest are powerset algebras and hence
(complete atomic) boolean algebras anyway.  Yet the coincidence itself
does make a difference: one-sorted domain semirings are easier to
formalise in interactive proof assistants and apply in program
verification and correctness.

%%%%%%%%%%%%%%%%%%%%%%%%%%%%%%%%%%%%%%%%%%%%%%%%%%%%%%%%

\section{Domain Axioms for Semirings}\label{S:domain-ops}

First we recall the two axiomatisations of domain semirings and their
relevant properties. To distinguish them, we call the first class,
introduced in~\cite{DesharnaisMS06}, \emph{test dioids with domain} and
the second one, introduced in~\cite{DesharnaisS11}, \emph{domain
  semirings}.

We assume familiarity with posets, lattices and semirings.  A
\emph{dioid}, in particular, is an idempotent semiring
$(S,+,\cdot,0,1)$, that is, $x+x=x$ holds for all $x\in S$. Its
additive monoid $(S,+,0)$ is then a semilattice ordered by
$x\le y\Leftrightarrow x+y=y$ and with least element $0$;
multiplication preserves $\le$ in both arguments.  (We generally omit
the $\cdot$ for multiplication.)

We write $S_1 = \{x\in S\mid x\le 1\}$ for the set of
\emph{subidentities} in $S$ and call $S$ \emph{bounded} if it has a
maximal element, $\top$.

We call a dioid $S$  \emph{full} if $S_1$ is a boolean algebra,
bounded by $0$ and $1$, with $+$ as sup, $\cdot$ as inf and an
operation $(\_)'$ of complementation that is defined only on $S_1$.

\begin{definition}[\cite{DesharnaisMS06}]
A \emph{test dioid} $(S,B)$ is a dioid $S$ that
contains a boolean subalgebra $B$ of $S_1$---the \emph{test algebra}
of $S$---with least element $0$, greatest element $1$, in which $+$
coincides with sup and that is closed under multiplication.
\end{definition}

Once again we write $(\_)'$ for complementation on $B$.

\begin{lemma}[\cite{DesharnaisMS06}]\label{P:test-meet}
  In every test dioid, multiplication of tests is their meet.
\end{lemma} 

\begin{lemma}[\cite{DesharnaisMS06}]\label{P:galois-aux}
  Let $(S,B)$ be a test dioid. Then, for all $x\in S$ and $p\in B$,
\begin{enumerate}
\item $x\le px \Leftrightarrow p'x= 0$, 
\item $x\le px \Leftrightarrow x\le p\top$ if $S$ is bounded.
\end{enumerate}
  \end{lemma}

\begin{definition}[\cite{DesharnaisMS06}]\label{D:domain1}
  A \emph{test dioid with predomain} is a test dioid $(S,B)$ with a
  \emph{predomain operation} $\dom:S\to B$ such that, for all $x\in S$
  and $p\in B$,
\begin{equation*}
  x \le \dom(x)x\qquad \text{ and } \qquad \dom(px)\le p.
\end{equation*}
It is a \emph{test dioid with domain} if it also satisfies, for
$x,y\in S$, the \emph{locality} axiom
  \begin{equation*}
    \dom(x\dom(y)) \le \dom(xy).
  \end{equation*}
\end{definition}

\emph{Weak locality} $\dom(xy)\le \dom(x\dom(y))$ already holds
in every test dioid with predomain. Thus
$\dom(x\dom(y)) = \dom(xy)$ in every test dioid with domain.

It is easy to check that binary relations and sets of paths satisfy
the axioms of test dioids with domain, and that $B=S_1$ in both
models.

\begin{lemma}[\cite{DesharnaisMS06}]\label{P:old-d-props} 
  In every test dioid $(S,B)$, the following statements are
  equivalent:
\begin{enumerate}
\item $(S,B,\dom)$ is a test dioid with predomain,
\item the map $\dom:S\to B$ on $(S,B)$ satisfies, for all $x\in S$ and
  $p\in B$, the \emph{least left absorption} property
\begin{equation}
  \label{eq:lla}
  \dom(x) \le p \Leftrightarrow x\le px,\tag{lla}
\end{equation}
\item in case $S$ is bounded, $\dom:S\to B$ on $(S,B)$ is, for all
  $x\in S$ and $p\in B$, the left adjoint in the adjunction
\begin{equation}
  \label{eq:d-adj}
  \dom(x) \le p \Leftrightarrow x\le p\top.\tag{d-adj}
\end{equation}
\end{enumerate}
\end{lemma}

Interestingly, test algebras of test dioids with domain cannot be
chosen ad libitum: they are formed by those subidentities that are
complemented relative to the multiplicative
unit~\cite{DesharnaisMS06}. This has the following consequences.

\begin{proposition}\label{P:old-d-semiring-bool}
  The test algebra $B$ of a test dioid with domain $(S,B,\dom)$ is the
  largest boolean subalgebra of $S_1$.
\end{proposition}

We write  $S_\dom=\{x\mid \dom(x) =x\}$
and $\dom(S)$ for the image of $S$ under $\dom$.

\begin{lemma}[\cite{DesharnaisS11}]\label{P:old-d-fix}
  Let $(S,B,\dom)$ be a test dioid with domain. Then $B=S_{\dom}=\dom(S)$. 
\end{lemma}

Next we turn to the second type of axiomatisation. 

\begin{definition}[\cite{DesharnaisS11}]
  A \emph{domain semiring} is a semiring $S$ with a map $\dom:S\to S$
  such that, for all $x,y\in S$ and with $\le $ defined as for
  dioids,
\begin{align}
 x &\le \dom(x)x,\tag{d1}\label{eq:d1}\\
\dom(x\dom(y)) &= \dom(xy),\tag{d2}\label{eq:d2}\\
 \dom(x) &\le 1,\tag{d3}\label{eq:d3}\\
\dom(0) &= 0,\tag{d4}\label{eq:d4}\\
\dom(x + y) &= \dom(x) + \dom(y).\tag{d5}\label{eq:d5}
\end{align}
\end{definition} 

Every domain semiring is a dioid:
$\dom(1) = \dom(1)1 = 1 + \dom(1)1 = 1 + \dom(1) = 1$, where the
second identity follows from (\ref{eq:d1}) and the last one from
(\ref{eq:d3}), therefore $1+ 1 = 1 + \dom(1)=1$ and finally
$x+x= x(1+1)=x$.  It follows that $\le$ is a partial order and that
axiom (\ref{eq:d1}) can be strengthened to $\dom(x) x = x$.

Once again it is straightforward to check that binary relations and
sets of paths form domain semirings.

In a domain semiring $S$, $\dom$ induces the domain algebra:
$\dom\circ \dom = \dom$ and therefore $S_\dom= \dom(S)$.  Moreover,
$(S_\dom,+,\cdot,0,1)$ forms a subsemiring of $S$, which is a bounded
distributive lattice with $+$ as binary sup, $\cdot$ as binary inf,
least element $0$ and greatest element $1$~\cite{DesharnaisS11}, but
not necessarily a boolean algebra.
\begin{example}[\cite{DesharnaisS11}]\label{ex:lattice1}
  The distributive lattice $0<a<1$ is a dioid with meet as
  multiplication, and a domain semiring with $\dom=\mathit{id}$ and
  therefore $S_\dom = S$.\qed
\end{example}

\begin{proposition}[\cite{DesharnaisS11}]\label{P:max-ds}
  The domain algebra of a domain semiring $S$ contains the largest boolean
  subalgebra of $S$ bounded by $0$ and $1$.
\end{proposition}

Axiom (\ref{eq:d5}) implies that $\dom$ is order preserving:
$x\le y \Rightarrow \dom(x)\le \dom(y)$. In addition,
$\dom(px) = p\dom(x)$ for all $p\in S_\dom$, $\dom(1)=1$, and
$\dom(\top) = 1$ if $S$ is bounded.  More importantly, \eqref{eq:lla}
can now be derived for all $p\in S_\dom$ (it need not hold for
$p\in S_1$)~\cite{DesharnaisS11}; it becomes an adjunction when $S$ is
bounded.

\begin{lemma}\label{P:d-adj}
  In any bounded domain semiring $S$, \eqref{eq:d-adj} holds for all
  $p\in S_\dom$.
\end{lemma}
\begin{proof}
$\dom(x)\le p$ implies 
$x = \dom(x) x \le px\le p \top$
and 
$\dom(x) \le \dom(p\top) = p\dom(\top) = p 1 = p$
follows from $x \le p\top$.
\end{proof}
% Thus $x\le px \Leftrightarrow x\le p\top$, for $p\in Q_\dom$, holds in
% any domain semiring with $\top$.

As mentioned in the introduction, an antidomain operation is needed to
make the bounded distributive lattice $S_\dom$ boolean. 

\begin{definition}[\cite{DesharnaisS11}]\label{D:antidom}
  An \emph{antidomain semiring} is a semiring $S$ with a an
  operation $\ad:S\to S$ such that, for all $x,y\in S$,
\begin{align*}
  \ad(x)x = 0,\qquad
 \ad(x) + \ad(\ad(x)) = 1,\qquad
\ad(xy) \le \ad(x\ad(\ad(y))).
\end{align*}
\end{definition}

Antidomain models boolean complementation in the domain algebra; the
domain operation can be defined as $\dom = \ad\circ \ad$ in any
antidomain semiring $S$.  The second and third antidomain axioms then
simplify to $\ad(x)+\dom(x)=1$ and $\ad(xy)\le \ad(x\dom(y))$.  The
domain algebra $S_\dom$ of $S$ is the maximal boolean subalgebra of
$S_1$, as in Proposition~\ref{P:old-d-semiring-bool}. This leads to
the following result.

\begin{lemma}[\cite{DesharnaisS11}]\label{P:ad-testd}
  Let $(S,\ad)$ be an antidomain semiring. Then $(S,S_\dom,\dom)$ is a
  test dioid with domain.
\end{lemma}

If the domain algebra $S_\dom$ of a domain semiring $S$ happens to be
a boolean algebra, it must be the maximal boolean subalgebra of $S_1$
by Proposition~\ref{P:max-ds}, so that $S$ is again a test dioid with
$B=S_\dom$.  Antidomain is then definable.
\begin{lemma}\label{P:dom-ad}
  Every domain semiring with  boolean domain algebra is an antidomain
  semiring.
\end{lemma}
\begin{proof}
  With $\ad = (\_)'\circ \dom$, the first antidomain axiom follows
  immediately from Lemma~\ref{P:galois-aux}(1); the remaining two
  axioms hold trivially.
\end{proof}
\begin{example}\label{ex:lattice2}
  In the dioid $0<a<1$ from Example~\ref{ex:lattice1},
  $\dom:0\mapsto 0,a\mapsto 1,1\mapsto 1$ defines another domain
  semiring with $S_\dom=\{0,1\}=B$. So $S_\dom\subset S_1$ is the
  maximal boolean subalgebra in $S_1$. In addition,
  $\ad:0\mapsto 1,a\mapsto 0,1\mapsto 0$ defines the corresponding
  antidomain semiring. Finally, this dioid is a test dioid by
  Lemma~\ref{P:ad-testd} and in fact a test dioid with domain in which
  $B=S_\dom\subset S_1$.\qed
\end{example}

As powerset algebras, relation and path domain semirings have of course
boolean domain algebras with complement $x'=1\cap \overline x$, where
$\overline x$ denotes complementation on the entire powerset
algebra. Both are therefore antidomain semirings, with the operations
shown in the introduction. 

We finish this section with an aside on fullness:\footnote{We are
  grateful to a reviewer for reminding us of this fact.}  While every
test dioid with domain and every antidomain semiring is full
whenever $S_\dom = S_1$ by Proposition~\ref{P:old-d-semiring-bool} and
Lemma~\ref{P:ad-testd}, in domain semirings, $S_\dom=S_1$ need not
imply that $S_\dom$ is boolean (Example~\ref{ex:lattice1}) and vice
versa (Example~\ref{ex:lattice2}).  A domain semiring $S$ is therefore
full precisely when $S_\dom$ is boolean and equal to $S_1$.

%%%%%%%%%%%%%%%%%%%%%%%%%%%%%%%%%%%%%%%%%%%%%%%%%%%%%%%

\section{Coincidence Result}\label{S:coincidence}

The results of Section~\ref{S:domain-ops} suggest that the two types
of domain semiring coincide when the underlying dioid is full.  We now
spell out this coincidence.

\begin{proposition}\label{P:tdd-ds}
  Let $(S,B,\dom)$ be a test dioid with domain. Then $(S,\dom)$ is a
  domain semiring with $S_\dom=B$ and an antidomain semiring with
  $\ad= (\_)'\circ \dom$.
\end{proposition}
\begin{proof}
  The domain semiring axioms are derivable in test dioids with
  domain~\cite{DesharnaisMS06}; the antidomain axioms follow by
  Lemma~\ref{P:dom-ad}. Moreover, $B$ is the maximal boolean
  subalgebra of $S_1$ by Proposition~\ref{P:old-d-semiring-bool}, and
  thus equal to $S_\dom$ by Proposition~\ref{P:max-ds} (alternatively
  Lemma~\ref{P:old-d-fix}).
\end{proof}

We know from Lemma~\ref{P:ad-testd} that every antidomain semiring is
a test dioid with domain. Hence, by Proposition~\ref{P:tdd-ds},
antidomain semirings and test dioids with domain are interdefinable
(see also~\cite{DesharnaisS11}).  For the other converse of
Proposition~\ref{P:tdd-ds} we consider full domain semirings $S$ where
$S_\dom=S_1$ is a boolean algebra by Proposition~\ref{P:max-ds}. These
are test dioids, hence \eqref{eq:lla} can be used to define domain.

\begin{corollary}\label{P:fds-tdd-cor}
  Let $S$ be a full dioid with map $\dom:S\to S$. Then \eqref{eq:lla}
  holds for all $x\in S$ and $p\in S_1$ if and only if the predomain
  axioms
\begin{equation*}
  x\le \dom(x) x\qquad \text{ and }\qquad \dom(px) \le p
\end{equation*}
 from Definition~\ref{D:domain1} hold for all $x\in S$ and
  $p\in S_1$.
\end{corollary}
\begin{proof}
  As $S$ is a test dioid with $B=S_1$, Lemma~\ref{P:old-d-props}(1)
  applies. 
\end{proof}
\begin{lemma}\label{P:fds-tdd-lem}
  Let $S$ be a full dioid with map $\dom:S\to S$ that satisfies
  \eqref{eq:lla} for all $x\in S$ and $p\in S_1$. Then $(S,S_1,\dom)$
  is a test dioid with predomain and $S_\dom=S_1$.
\end{lemma}
\begin{proof}
  $S$ is a test dioid with predomain by
  Corollary~\ref{P:fds-tdd-cor}. $S_\dom\subseteq S_1$ because
  $\dom(x) \le 1$ in any test dioid with
  predomain~\cite{DesharnaisMS06}. $S_1\subseteq S_\dom$ because $p\le 1$
  implies $p=\dom(p)p \le \dom(p)$ and $\dom(p) \le p$ because $pp=p$,
  using \eqref{eq:lla}. 
\end{proof}

\begin{proposition}\label{P:fds-tdd}
  Let $(S,\dom)$ be a full domain semiring. 
Then $(S,S_\dom,\dom)$ is a test
  dioid with domain. 
\end{proposition}
\begin{proof}
  If $(S,\dom)$ is a full domain semiring, then \eqref{eq:lla} is
  derivable and locality holds. Then $(S,S_\dom,\dom)$ is a test dioid with
  predomain by Lemma~\ref{P:fds-tdd-lem} and therefore a test dioid
  with domain because of locality.
\end{proof}

Our coincidence result, through which the two types of domain
semirings are united, then follows easily from
Propositions~\ref{P:tdd-ds} and \ref{P:fds-tdd}.

\begin{theorem}\label{P:main-thm}
  A full test dioid is a test dioid with domain if and only if it is a
  domain semiring. 
\end{theorem}

On full dioids, domain can therefore be axiomatised either
equationally by the domain semiring axioms or those of test dioids
with domain, or alternatively by \eqref{eq:lla} and locality. The
domain algebras of relation and paths domain semirings, in particular,
are full.

In any dioid, hence in particular any domain semiring, fullness can be
enforced, for instance, by requiring that every $p\in S_1$ be
\emph{complemented} within $S_1$, that is, there exists an element
$q\in S _1$ such that $p+q=1$ and $qp=0$.  It then follows that $S_1$
is a boolean algebra~\cite{DesharnaisS11}.

Alternatively, in any test dioid with domain or any antidomain
semiring, $S_\dom=S_1$ whenever $x\le 1 \Rightarrow \dom(x)=x$, for
all $x\in S$. Yet Example~\ref{ex:lattice2} shows that this
implication does not suffice to make $S_\dom$ boolean in arbitrary
domain semirings.

Finally, locality need not hold in full test dioids that satisfy
\eqref{eq:lla}.

\begin{example}\label{ex:locality}
  Consider the full test dioid with $S=\{0,1,a,\top\}$ in which $a$
  and $1$ are incomparable with respect to $\le$, $aa =0$,
  multiplication is defined by $a\top = \top a = a$ and
  $\top\top = \top$, and $\dom$ maps $0$ to $0$ and every other
  element to $1$. Then \eqref{eq:lla} holds, but
  $\dom(a\dom(a)) = \dom(a1) = \dom(a) = 1 > 0 = \dom(0) =
  \dom(aa)$.\qed
\end{example}

%%%%%%%%%%%%%%%%%%%%%%%%%%%%%%%%%%%%%%%%%%%%%%%%%%%%%%%%%%

\section{Examples}

The restriction to full test dioids is natural for concrete powerset
algebras, like the relation and path algebras mentioned. It is
captured abstractly, for instance, by boolean monoids and quantales.

A \emph{boolean monoid}~\cite{DesharnaisMS06} is a structure
$(S,+,\sqcap,\cdot,\overline{\phantom{x}},0,1,\top)$ such that
$(S,+,\cdot,0,1)$ is a semiring and
$(S,+,\sqcap,\overline{\phantom{x}},0,\top)$ a boolean algebra.  As
all sups, infs and multiplications of subidentities stay below $1$,
every boolean monoid is a full bounded dioid; boolean complementation
on $S_1$ is given by $p' = 1\sqcap \overline{p}$ for all $p\in S_1$.

Domain can now be axiomatised as an endofunction, either equationally
using the domain semiring or test dioid with domain axioms, or by the
adjunction \eqref{eq:d-adj} and locality, as in
Section~\ref{S:coincidence}.  Once again, the antidomain operation
$\ad$ is complementation on $S_1$. Theorem~\ref{P:main-thm} has the
following instance.

\begin{corollary}\label{P:main-thm-cor1}
  A boolean monoid is a test dioid with domain if and only if it is a
  domain semiring.  
\end{corollary}

Quantales capture the presence of arbitrary sups and infs in powerset
algebras more faithfully. Formally, a \emph{quantale}
$(Q,\le,\cdot, 1)$ is a complete lattice $(Q,\le)$ and a monoid
$(Q,\cdot,1)$ such that composition preserves all sups in its first
and second argument. We write $\Sup$ for the sup
and $\Inf$ for the inf operator.  We also write $0= \Inf Q$ for the
least and $\top=\Sup Q$ for the greatest element of $Q$, and $\lor$
and $\land$ for binary sups and infs.

A quantale is \emph{boolean} if its complete lattice is a boolean
algebra.  Every boolean quantale is obviously a boolean monoid, and
every finite boolean monoid a boolean quantale. If $Q$ is a boolean
quantale, then $Q_1$ forms even a complete boolean algebra. In boolean
quantales, predomain, domain and antidomain operations can therefore
be axiomatised like in boolean monoids, and we obtain another instance
of Theorem~\ref{P:main-thm}, analogous to
Corollary~\ref{P:main-thm-cor1}, simply by replacing ``boolean
monoid'' with ``boolean quantale''. 

As for domain semirings, $Q_\dom$ need neither be full nor boolean in
an arbitrary domain quantale: the dioids in Examples~\ref{ex:lattice1}
and \ref{ex:lattice2} are defined over finite semilattices and hence
complete lattices. They are therefore quantales. In this case, the
identity $\dom(x\land 1) = x\land 1$ forces $Q_\dom = Q_1$, because
this inequality implies $\dom(x)=x$ for all $x\le 1$, and in fact a
domain semiring with a meet operation suffices for the
proof.\footnote{Again we owe this observation to a reviewer.} In
antidomain quantales, this identity thus implies fullness. Whether or
how the fullness could be forced equationally in arbitrary domain
semirings or antidomain semirings is left open.

%%%%%%%%%%%%%%%%%%%%%%%%%%%%%%%%%%%%%%%%%%%%%%%%%%%

\section{Domain Quantales}

Some loose ends remain to be tied together in this note as well: 
\begin{itemize}
\item Does the interaction of domain with arbitrary sups and infs in
  quantales require additional axioms?
\item Why has domain not been axiomatised explicitly using
the adjunction \eqref{eq:d-adj}, at least for boolean quantales?  
\item And
why has domain in boolean monoids or quantales not been axiomatised
explicitly by $\dom(x) = 1 \land x\top$, as in relation algebra? 
\end{itemize}
This section answers these questions.

First, we consider the domain semiring axioms in arbitrary quantales
and argue that additional sup and inf axioms are unnecessary.

\begin{definition}
  A \emph{domain quantale} is a quantale that is also a domain
  semiring.
\end{definition}

As every quantale is a bounded dioid, the adjunction \eqref{eq:d-adj}
holds for every $p\in Q_d$.  In addition, domain interacts with sups
and infs as follows.

\begin{lemma}\label{P:d-props-quantale}
  In every domain quantale,
  \begin{enumerate}
  \item $\dom(\Sup X) = \Sup\dom(X)$,
  \item $\dom(\Inf X)\le \Inf\dom(X)$,
%\item $x \Inf Y\le \Inf x Y$ and  $(\Inf  X) y \le \Inf Xy$,
\item $\dom(x)( \Inf Y) = \Inf \dom(x)Y$ for all $Y\neq \emptyset$.
\end{enumerate}
\end{lemma}
\begin{proof}~
\begin{enumerate}
\item $\dom$ is a left adjoint by Lemma~\ref{P:d-adj} and therefore
  sup-preserving.  Sups over $X$ are taken in $Q$; those over
  $\dom(X)$ in $Q_\dom$.
\item
  $\left(\forall x\in X.\ \Inf X \le x\right) \Rightarrow \left(\forall x\in X.\ \dom(
  \Inf  X) \le  \dom( x)\right) \Leftrightarrow \dom(\Inf X) \le \Inf\dom(X)$.
\item Every $y\in Y\neq \emptyset$ satisfies
  \begin{equation*}
\dom\left(\Inf \dom(x) Y\right)\le \dom(\dom(x)y) = \dom(x)\dom(y)
  \le \dom(x)
\end{equation*}
  and therefore
  \begin{equation*}
\Inf \dom(x) Y = \dom\left(\Inf\dom(x) Y\right)\left( \Inf\dom(x) Y\right)\le
  \dom(x)\left( \Inf \dom(x) Y\right) \le \dom(x)\left( \Inf Y\right).
\end{equation*}
  The converse inequality holds because $x( \Inf Y)\le \Inf x Y$ in any
  quantale.\qedhere
\end{enumerate}
\end{proof}

If $Y=\emptyset$ in part (3) of the lemma, then
$\dom(x)( \Inf Y) = \dom(x) \top$ need not be equal to
\begin{equation*}
\top = \Inf \emptyset = \Inf \dom(x) Y.
\end{equation*}
In the quantale of binary relations over the set $\{a,b\}$, for
instance, $R=\{(a,a)\}$, satisfies $\dom(R)= R$ and
\begin{equation*}
\dom(R) \top = \left\{(a,a)\right\}\cdot \{(a,a),(a,b),(b,a),(b,b)\}
= \{(a,a),(a,b)\} \subset \top.
\end{equation*}

Moreover, part (1) of the lemma implies that the domain algebra
$Q_\dom$ is a \emph{complete} distributive lattice:
$\dom (\Sup\dom(X)) = \Sup \dom(X)$ holds for all $X\subseteq Q$, so
that any sup of domain elements is again a domain element. Yet the
sups and infs in $Q_\dom$ need not coincide with those in $Q$.

Second, the adjunction $\dom(x)\le p \Leftrightarrow x \le p\top$
holds for all $p\in Q_1$ in a boolean quantale $Q$. General properties
of adjunctions then imply that, for all $x\in Q$,
\begin{equation*}
\dom(x) = \Inf\{p\in Q_1\mid x\le p\top\}.
\end{equation*}
Lemma~\ref{P:fds-tdd-lem}, in turn, guarantees that this identity
defines predomain explicitly on boolean quantales. Yet
Example~\ref{ex:locality} rules out that it defines domain: the full
test dioid from this example is, in fact, a boolean quantale; it
satisfies \eqref{eq:lla} and thus \eqref{eq:d-adj}, but violates the
locality axiom of domain quantales.

Finally, we give two reasons why the relation-algebraic identity 
\begin{equation*}
\dom(x) = 1 \land x \top
\end{equation*}
cannot replace the domain axioms in boolean monoids and quantales.

It is too weak: In the boolean quantale $\{\bot,1,a,\top\}$ with $1$
and $a$ incomparable and multiplication defined by $\top\top=\top$ and
$aa=a\top =\top a = a$, it holds that $d(a) = \bot$ (when defined by
$\dom(x) = 1 \land x \top$), yet $\dom(a)a = \bot a=\bot < a$.
Therefore $\dom(x)x = x$ is not derivable from
$\dom(x) = 1 \land x\top$ even in boolean quantales.

It is too restrictive: although $\dom(x)= 1\land x\top$ obviously
holds in the quantale of binary relations, it fails, for instance, in
the quantale formed by the sets of (finite) paths over a digraph
$\sigma,\tau:E\to V$ mentioned in the introduction. Recall that the domain elements
of a set $P$ of paths are a subset of $V$ given by the sources of the
these paths. It is then obvious that $V\cap P\top=\emptyset$ unless
$P$ contains a path of length one and
$\dom(P) =\emptyset \Leftrightarrow P=\emptyset$, so that
$\dom(P) = V\cap P\top$ fails for any $P$ in which all paths have
length greater than $1$.

This type of argument applies to all powerset quantales in which the
composition of underlying objects (here: paths) is generally
length-increasing and the quantalic unit and domain elements are
formed by fixed-length objects.

\vspace{\baselineskip}

\noindent\textbf{Acknowledgments:} We would like to thank the
journal reviewers for their very insightful comments.

%%%%%%%%%%%%%%%%%%%%%%%%%%%%%%%%%%%%%%%%%%%%%%%%%%%%%%%

\bibliographystyle{alpha}
\bibliography{kad-united}

\end{document}